\documentclass{birkjour}

\usepackage{todonotes,soul}
\usepackage{tikz}
\usetikzlibrary{matrix,arrows,shapes,decorations.pathmorphing}
\usepackage{bbding} 

\usepackage{extarrows}

\DeclareMathAlphabet{\mathcal}{OMS}{cmsy}{m}{n}

\makeatletter
\g@addto@macro\bfseries{\boldmath}
\makeatother

\DeclareMathAlphabet\BEuScript{U}{eus}{b}{n}
\newcommand{\Mbb}{{\BEuScript{M}}}
\usepackage{amsmath,amssymb,amscd,mathrsfs}

 \newtheorem{thm}{Theorem}[section]

 \newtheorem{axiom}[thm]{Axiom}
 \theoremstyle{definition}
 \newtheorem{definition}[thm]{Definition}
 \theoremstyle{remark}

 \numberwithin{equation}{section}

\newcommand{\CoinX}[1]{C_0^\infty({#1})}

\newcommand{\Af}{{\mathscr{A}}}

\newcommand{\Sf}{{\mathscr{S}}}

\newcommand{\supp}{{\rm supp}\,}

\newcommand{\CC}{{\mathbb C}}
\newcommand{\RR}{{\mathbb R}}

\newcommand{\Lc}{{\mathcal{L}}} 
\newcommand{\Mc}{{\mathcal{M}}}

\newcommand{\Zc}{{\mathcal{Z}}}
\newcommand{\OO}{{\mathcal{O}}}

\newcommand{\Ib}{{\boldsymbol{I}}}
\newcommand{\Lb}{{\boldsymbol{L}}}
\newcommand{\Mb}{{\boldsymbol{M}}}
\newcommand{\Nb}{{\boldsymbol{N}}}

\newcommand{\nto}{\stackrel{.}{\to}}

\newcommand{\id}{{\rm id}}
\newcommand{\II}{{\mathbf{1}}}

\newcommand{\FLoc}{{\sf FLoc}}
\newcommand{\Loc}{{\sf Loc}}
\newcommand{\SpLoc}{{\sf SpinLoc}}

\newcommand{\Alg}{{\sf Alg}}
\newcommand{\CAlg}{{\sf C^*\hbox{-}Alg}}

\newcommand{\Bf}{{\mathscr B}}

\newcommand{\Df}{{\mathscr D}}

\newcommand{\Lf}{{\mathscr L}}

\newcommand{\Tf}{{\mathscr T}}

\DeclareMathOperator{\Aut}{Aut}

\DeclareMathOperator{\Fld}{Fld}

\newcommand{\dvol}{d{\rm vol}}

\DeclareMathOperator{\rce}{rce}
\newcommand{\kin}{\text{kin}}

\DeclareMathOperator{\SL}{SL}
 
\begin{document}

\title[Spin-statistics connection in curved spacetimes]{On the spin-statistics connection in curved spacetimes}
\author[CJ Fewster]{Christopher J Fewster}
\address{Department of Mathematics, University of York, Heslington, York YO10 5DD, United Kingdom}
\email{chris.fewster@york.ac.uk}
\date{\today}

\begin{abstract}
The connection between spin and statistics is examined in the context of locally covariant
quantum field theory. A generalization is proposed in which locally covariant theories are defined as
functors from a category of framed spacetimes to a category of $*$-algebras. 
This allows for a more operational
description of theories with spin, and for the derivation of a more general version of the spin-statistics
connection in curved spacetimes than previously available. The proof
involves a ``rigidity argument'' that is also applied in the standard setting of 
locally covariant quantum field theory to show how properties such as 
Einstein causality can be transferred from Minkowski spacetime to general curved spacetimes. 
\end{abstract}
\maketitle

\newcommand{\alert}[1]{{\bf #1}}
\newcommand{\pause}{}

\section{Introduction}

\begin{quotation}
\emph{In conclusion we wish to state, that according
to our opinion the connection between spin and
statistics is one of the most important applications
of the special relativity theory.}\\ W. Pauli, in~\cite{Pauli:1940}.
\end{quotation}
 
It is an empirical fact that observed elementary particles are either
bosons of integer spin, or fermions of half-integer spin. 
Explanations of this connection between spin and statistics 
have been sought since the early days of quantum field theory. 
Fierz~\cite{Fierz:1939} and Pauli~\cite{Pauli:1940} investigated
the issue in free field theories, setting in train a number of
progressively more general results. The rigorous proof of a
connection between spin and statistics was an early and major
achievement of the axiomatic Wightman framework; see
\cite{Burgoyne:1958, LuedersZumino:1958} and the classic presentation in ~\cite{StreaterWightman}. Similarly, general results
have been proved in the Haag--Kastler framework~\cite{Haag}, 
for example, \cite{Epstein:1967,DHRiv,GuidoLongo:1995}. In these more
algebraic settings, statistics is not tied to the properties of particular fields, 
but rather understood in terms of the graded commutativity of local
algebras corresponding to spacelike-separated regions~\cite{Epstein:1967}, or the properties
of superselection sectors \cite{DHRiv,GuidoLongo:1995}.
 
Nonetheless, the theoretical account of the spin-statistics connection is subtle and even
fragile. Nonrelativistic models of quantum field theory are not
bound by it, and as Pauli observed~\cite{Pauli:1940}, one may
impose bosonic statistics on a Dirac field at the cost of
sacrificing positivity of the Hamiltonian. Ghost fields introduced
in gauge theories violate the connection, but also involve indefinite
inner products. The rigorous proofs therefore rely on Hilbert space
positivity and energy positivity. Moreover, they 
make essential use of the 
Poincar\'e symmetry group and  its complex extension
together with analyticity properties of the vacuum $n$-point
functions. The spin-statistics connection observed in nature, 
however, occurs in a spacetime which is not Minkowski space
and indeed has no geometrical symmetries. There is neither a global 
notion of energy positivity (or, more properly, the spectrum condition) nor do we expect $n$-point functions in typical states of interest
on generic spacetimes to have analytic  extensions. Thus the general
proofs mentioned have no traction and it is far from clear how they can be generalized: 
\emph{a priori} it is quite conceivable
that the theoretical spin-statistics connection is an accident of special relativity that is broken in passing to the curved spacetimes of general
relativity.

Indeed, for many years, work on the spin-statistics
connection in curved spacetimes was restricted
to demonstrations that free models become inconsistent
on general spacetimes if equipped with the wrong
statistics (e.g., imposing anticommutation relations
on a scalar field)  \cite{Wald_Smatrix:1979,ParkerWang:1989} 
unless some other important
property such as positivity is sacrificed~\cite{HiguchiParkerWang:1990}.

The breakthrough was made by Verch \cite{Verch01},
who established a general spin-statistics theorem for 
theories defined on each spacetime by a single field which, in particular,
obeys Wightman axioms in Minkowski space. Together with \cite{Ho&Wa02}, this paper was responsible for laying down many of the 
foundations of what has become the locally covariant framework for QFT in curved spacetimes~\cite{BrFrVe03}. 
Verch's assumptions allow certain properties of the theory on 
one spacetime to be deduced from its properties on another,
provided the spacetimes are suitably related by restrictions
or deformations of the metric. In particular, the spin-statistics
connection is proved by noting that if it were violated in any one spacetime, it 
would be violated in Minkowski space, contradicting the classic spin-statistics theorem. 

Nonetheless, there are good reasons to revisit the 
spin-statistics connection in curved spacetime.
First, as a matter of principle, one would like to gain a
better undestanding of why spin is the correct concept
to investigate in curved spacetime, given the lack of
the rotational symmetries that are so closely bound up
with the description of spin in Minkowski space. 
A second, related, point is that \cite{Verch01} describes spinor fields 
as sections of various bundles associated to the spin bundle. 
While this is conventional wisdom in QFT in CST,  it has the effect of
basing the discussion on geometric structures that are, in part, unobservable.
This is not a great hindrance if the aim is to discuss a particular model  such as
the Dirac field. However, we wish to understand the spin-statistics
connection for general theories, without necessarily basing the description
on fields at all. With that goal in mind, one needs a more fundamental
starting point that avoids the insertion of spin by hand. Third, the result proved in \cite{Verch01} is confined to theories in which the algebra
in each spacetime is generated by a single field, and the argument
is indirect in parts. The purpose of this contribution is to sketch a new and
operationally well-motivated perspective on the spin-statistics connection 
in which spin emerges as a natural concept in curved spacetimes,
and which leads to a more general and direct proof of
the connection. In particular, there is no longer any need to
describe the theory in terms of one or more fields. 
Full details will appear shortly~\cite{Few_spinstats}.

The key ideas are (a) a formalisation of the reasoning underlying \cite{Verch01} 
as a `rigidity argument', and (b) a generalization of locally covariant QFT based on a category of spacetimes with global coframes
(i.e.,  a `rods and clocks' account of  spacetime measurements). As in \cite{Verch01}
the goal is to prove that a spin-statistics connection in curved spacetime is implied by the standard results holding in Minkowski space; 
however, the proof becomes quite streamlined in the new formulation.
We begin by describing the standard version of locally covariant QFT, describing the 
rigidity argument and some of its other applications in that context, before moving
to discussion of framed spacetimes and the spin-statistics theorem.

\section {Locally covariant QFT}

Locally covariant QFT is a general framework for QFT in curved spacetimes, due to Brunetti, Fredenhagen and Verch (BFV) \cite{BrFrVe03},
which comprises three main assumptions. The first is the assertion that any quantum field theory respecting locality and covariance
can be described a covariant functor $\Af:\Loc\to \Alg$ from the category of globally hyperbolic spacetimes $\Loc$ to a category $\Alg$ of unital $*$-algebras.\footnote{Other target categories are often used, e.g., the unital $C^*$-algebra category $\CAlg$, and other types of physical theory can be 
accommodated by making yet other choices.}

This assumption already contains a lot of information and we shall unpack it in stages,
beginning with the spacetimes. Object of $\Loc$ are oriented and time-oriented globally hyperbolic spacetimes (of fixed dimension $n$, which we will take to be four) and with finitely many components.\footnote{It is convenient to describe 
the orientation by means of a connected component of the set of nonvanishing 
$n$-forms, and likewise to describe the time-orientation by means of connected component of the set of nonvanishing timelike $1$-form fields.
Our signature convention throughout is $+-\cdots-$.}
Morphisms between spacetimes in $\Loc$ are \emph{hyperbolic embeddings}, i.e., isometric embeddings preserving time and space orientations with causally convex image. 

The category $\Alg$
has objects that are unital $*$-algebras, with morphisms that are injective, unit-preserving $*$-homomorphisms. 
The functoriality condition requires that
the theory assigns an object $\Af(\Mb)$ of $\Alg$ to each spacetime $\Mb$ of $\Loc$,
and, furthermore, that each hyperbolic embedding of spacetimes $\psi:\Mb\to\Nb$
is mirrored by an embedding of the corresponding algebras   
$\Af(\psi):\Af(\Mb)\to \Af(\Nb)$, 
such that 
\begin{equation}
\Af(\id_\Mb) = \id_{\Af(\Mb)} \quad\text{and} \quad \Af(\varphi\circ\psi) = \Af(\varphi)\circ \Af(\psi)
\end{equation}
for all composable embeddings $\varphi$ and $\psi$.

Despite its somewhat formal expression, this assumption is well-motiv{-}ated from an
operational viewpoint\footnote{For a discussion 
of how the framework can be motivated on operational grounds (and
as an expression of `ignorance principles')  
see \cite{FewsterRegensburg,Few_Chicheley:2015}.} and provides a natural generalization of the Haag--Kastler--Araki
axiomatic description of quantum field theory in Minkowski space. Indeed, 
as emphasized by BFV, this single assumption already contains 
several distinct assumptions of the Minkowski framework.  

The next ingredient in the BFV framework is the \emph{kinematic net} indexed by $\OO(\Mb)$, the set of all open causally convex subsets of $\Mb$ with finitely many connected components.
Each $O\in\OO(\Mb)$ can be regarded as a spacetime $\Mb|_O$ in its own right, by restricting the 
causal and metric structures of $\Mb$ to $O$, whereupon the inclusion map of $O$ into the underlying
manifold of $\Mb$ induces a $\Loc$ morphism $\iota_{O}:\Mb|_O\to\Mb$ (see Fig.~\ref{fig:inclusion}). The theory $\Af$
therefore assigns an algebra $\Af(\Mb|_O)$ and an embedding of this algebra into $\Af(\Mb)$, 
and we define the kinematic subalgebra to be the image
\begin{equation}
\Af^\kin(\Mb;O) := \Af(\iota_{O})(\Af(\Mb|_O)).
\end{equation}
As mentioned above, the net $O\mapsto\Af^\kin(\Mb;O)$ is the appropriate
generalization of the net of local observables studied in Minkowski space AQFT.  
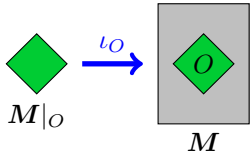
\begin{figure}
\begin{center}
\begin{tikzpicture}[scale=0.4]
\definecolor{Green}{rgb}{0,.80,.20}
\draw[fill=lightgray] (-7,0) ++(-1.5,0) -- ++(3,0) -- ++(0,4) --
++(-3,0) -- cycle;
\draw[fill=Green] (-7,2) +(-1,0) -- +(0,-1) -- +(1,0) -- +(0,1) -- cycle;
\node[anchor=north] at (-7,0) {$\Mb$};
\draw[fill=Green] (-12.5,2) +(-1,0) -- +(0,-1) -- +(1,0) -- +(0,1) -- cycle;
\node at (-7,2) {$O$};
\node[anchor=north] at (-12.5,1) {$\Mb|_O$};
\draw[line width=2pt,->,color=blue] (-11,2) -- (-9,2) node[pos=0.5,above] {$\iota_O$};
\end{tikzpicture}
\end{center}
\caption{Schematic illustration of the kinematic net.}\label{fig:inclusion}
\end{figure}
Some properties are automatic. For instance, the kinematic algebras are covariantly defined, in the sense that
\begin{equation}\label{eq:kincov}
\Af^\kin(\Nb;\psi(O))=\Af(\psi)(\Af^\kin(\Mb;O))
\end{equation}
for all morphisms $\psi:\Mb\to\Nb$ and all nonempty $O\in\OO(\Mb)$. 
This is an immediate consequence of the definitions above and functoriality
of $\Af$. Similarly spacetime symmetries of $\Mb$ are realised as automorphisms
of the kinematic net in a natural way. 

It is usual to assume two additional properties. 
First,  the theory obeys \emph{Einstein causality} if, for all causally disjoint
$O_1,O_2\in\OO(\Mb)$ (i.e., no causal curve joins $O_1$ to $O_2$), the
corresponding kinematic algebras commute elementwise. 
Second, $\Af$ is said to have the \emph{timeslice property} if it 
maps every \emph{Cauchy morphism}, i.e., a morphism
whose image contains a Cauchy surface of the ambient spacetime, to an isomorphism
in $\Alg$.  This assumption encodes the dynamics of the theory and plays an important role
in allowing the instantiations of $\Af$ on different spacetimes to be related.
In fact, any two spacetimes $\Mb$ and $\Nb$ in $\Loc$ can be linked by a chain of Cauchy morphisms
if and only if their Cauchy surfaces are related by an orientation-preserving
diffeomorphism (see \cite[Prop.~2.4]{FewVer:dynloc_theory}, which builds on an
older argument of Fulling, Narcowich and Wald~\cite{FullingNarcowichWald}).
The construction used is shown schematically in Fig.~\ref{fig:deform}: the
main point is the construction of the interpolating spacetime $\Ib$
that `looks like' $\Nb$ in its past and $\Mb$ in its future.
\begin{figure}
\begin{center}
\begin{tikzpicture}[scale=0.6]
\definecolor{Green}{rgb}{0,.80,.20}
\definecolor{Gold}{rgb}{.93,.82,.24}
\definecolor{Orange}{rgb}{1,0.5,0}
\draw[fill=Gold] (-7,2) +(-1,0) -- +(0,-1) -- +(1,0) -- +(0,1) -- cycle;
\draw[fill=Gold] (-3,2) +(-1,0) -- +(0,1) -- +(1,0) ..  controls +(0,0.5) .. +(-1,0);
\draw[fill=Gold] (1,2) +(-1,0) -- +(0,1) -- +(1,0) ..  controls +(0,0.5) .. +(-1,0);
\draw[fill=Green] (1,2) +(-1,0) -- +(0,-1) -- +(1,0) ..  controls +(0,-0.5) .. +(-1,0);
\shadedraw[top color =Gold,bottom color=Green] (1,2) +(-1,0) ..  controls +(0,0.5) .. +(1,0) ..  controls +(0,-0.5) .. +(-1,0);
\draw[fill=Green] (5,2) +(-1,0) -- +(0,-1) -- +(1,0) ..  controls +(0,-0.5) .. +(-1,0);
\draw[fill=Green] (9,2) +(-1,0) -- +(0,-1) -- +(1,0) -- +(0,1) -- cycle;
\draw[color=blue,line width=4pt,->] (-4,2.5) -- (-6,2.5);
\draw[color=blue,line width=4pt,<-] (0,2.5) -- (-2,2.5);
\draw[color=blue,line width=4pt,->] (4,1.5) -- (2,1.5);
\draw[color=blue,line width=4pt,->] (6,1.5) -- (8,1.5);
\node[anchor=north] at (-7,1) {$\Mb$};
\node[anchor=north] at (1,1) {$\Ib$};
\node[anchor=north] at (9,1) {$\Nb$};
\end{tikzpicture}
\end{center}
\caption{Schematic representation of spacetime deformation.} \label{fig:deform}
\end{figure}
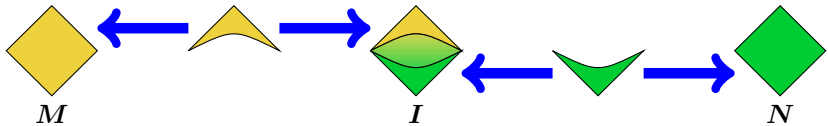
The assumption that $\Af$ has the timeslice property entails the existence
of an isomorphism between $\Af(\Mb)$ and $\Af(\Nb)$; indeed, there
are many such isomorphisms, because there is considerable freedom in the
choice of interpolating spacetime, none of which can be regarded as canonical. 

The assumptions just stated are satisfied by simple models, such
as the free Klein--Gordon field~\cite{BrFrVe03}, and, importantly, for
perturbatively constructed models of a scalar field with self-interaction \cite{BrFr2000,Ho&Wa01,Ho&Wa02}. In order to be self-contained, we briefly describe the free theory corresponding to the minimally coupled Klein--Gordon theory,  with field equation $(\Box_\Mb+m^2)\phi=0$:  in each spacetime $\Mb\in\Loc$, one defines a unital $*$-algebra $\Af(\Mb)$
with generators $\Phi_\Mb(f)$ (`smeared fields') labelled by test functions $f\in\CoinX{\Mb}$ and subject to the following relations:
\begin{itemize}
\item $f\mapsto\Phi_\Mb(f)$ is linear
\item $\Phi_\Mb(f)^*=\Phi_\Mb(\overline{f})$
\item $\Phi_\Mb((\Box_\Mb+m^2)f) = 0$
\item $[\Phi_\Mb(f),\Phi_\Mb(f')] = iE_\Mb(f,f')\II_{\Af(\Mb)}$
\end{itemize}
where 
\begin{equation}
E_\Mb(f,f')= \int_\Mb f(p) \left((E_\Mb^--E_\Mb^+)f'\right)(p)\dvol_\Mb(p)
\end{equation}
is constructed from the advanced ($-$) and retarded ($+$) Green operators
(which obey $\supp(E_\Mb^\pm f)\subset J^\pm_\Mb(\supp f)$). This defines the objects of the theory; for the morphisms,
any hyperbolic embedding $\psi:\Mb\to\Nb$ determines a unique morphism
$\Af(\psi):\Af(\Mb) \to\Af(\Nb)$ with the property
\begin{equation}
\Af(\psi)\Phi_\Mb(f)=\Phi_\Nb(\psi_* f) \qquad (f\in\CoinX{\Mb})
\end{equation}
The proof that $\Af(\psi)$ is well-defined as a morphism of $\Alg$ relies on the
properties of globally hyperbolic spacetimes, the definition of hyperbolic embeddings,
and some algebraic properties of the algebras $\Af(\Mb)$ [notably, that they are simple].
 
Our discussion will use two more features of the general structure.
First, the field content is described by natural transformations
between the functor $\Df$ assigning the space of test functions to spacetimes and the functor $\Af$ (suppressing a forgetful functor)
\cite{BrFrVe03}.
This means that to each spacetime $\Mb\in\Loc$ there is a 
map $\Phi_\Mb:C_0^\infty(\Mb)\to\Af(\Mb)$, so that
\begin{equation}
\Af(\psi)\Phi_\Mb(f) = \Phi_\Nb(\psi_* f)
\end{equation}
holds for each $\psi:\Mb\to\Nb$, where $\psi_*$ is push-forward.
The prototypical example is the Klein--Gordon field, of course.
The description just given applies to scalar fields;  fields
of other tensorial types may be incorporated by suitable alternatives choices
of $\Df$. As will be discussed later, spinorial fields require
a modification of the category $\Loc$.

Second, natural transformations may also be used to compare locally covariant theories. A natural $\eta:\Af\nto\Bf$ is interpreted
as an embedding of $\Af$ as a subtheory of $\Bf$, while
a natural isomorphism indicates that the theories are physically
equivalent~\cite{BrFrVe03,FewVer:dynloc_theory}. Naturality
requires that to each $\Mb\in\Loc$ there is a morphism
$\zeta_\Mb:\Af(\Mb)\to
\Bf(\Mb)$ 
\begin{equation}
\zeta_\Nb\circ\Af(\psi) = \Bf(\psi)\circ\zeta _\Mb
\end{equation}
for each morphism  $\psi:\Mb\to\Nb$. The interpretation
of $\eta$ as a subtheory embedding can be justified on several
grounds -- see~\cite{FewVer:dynloc_theory}. 

The equivalences of $\Af$ with itself form the group $\Aut(\Af)$
of automorphisms of the functor. This has a nice physical
interpretation: it is the global gauge group \cite{Fewster:gauge}.

Locally covariant QFT is not merely an elegant formalism
for rephrasing known results and models. It has led to new
departures in the description of QFT in curved spacetimes.
These can be divided into those that are model-independent
and those that are specific to particular theories. Those
of the former type include the spin-statistics connection \cite{Verch01}; the introduction of the relative Cauchy evolution and 
intrinsic understanding of the stress-energy tensor \cite{BrFrVe03};
an analogue of the 
Reeh--Schlieder theorem~\cite{Sanders_ReehSchlieder,Few_split:2015}
and the split property \cite{Few_split:2015}; new approaches
to superselection theory~\cite{Br&Ru05,BrunettiRuzzi_topsect}
and the understanding of global gauge transformations~\cite{Fewster:gauge};
a no-go theorem for preferred states~\cite{FewVer:dynloc_theory},
and a discussion of how one can capture the idea that a theory
describes  `the same physics in all spacetimes' \cite{FewVer:dynloc_theory}. 
Model-specific applications include, above all, the 
perturbative construction of interacting models  \cite{BrFr2000,Ho&Wa01,Ho&Wa02}, including those with 
gauge symmetries \cite{Hollands:2008,FreRej_BVqft:2012}. However, there are
also applications to the theory of Quantum Energy Inequalities  \cite{Few&Pfen06, Marecki:2006,Fewster2007} and  cosmology  \cite{DapFrePin2008,DegVer2010,VerchRegensburg}.

\section{A rigidity argument}\label{sect:rigidity}

The framework of local covariance appears quite loose, but 
in fact the descriptions of the theory in 
different spacetimes are surprisingly tightly related.    
There are various interesting properties 
which, if they hold in Minkowski space, must also hold in general spacetimes.  
This will apply in particular to the spin--statistics connection;
as a warm-up, let us see how such arguments can be used
in the context of Einstein causality, temporarily relaxing our
assertion of this property as an axiom.

For $\Mb\in\Loc$, let $\OO^{(2)}(\Mb)$ be the set of ordered pairs 
of spacelike separated open globally hyperbolic subsets of $\Mb$.
For any such pair $\langle O_1,O_2\rangle\in\OO^{(2)}(\Mb)$, let 
$P_\Mb(O_1,O_2)$ be true if
$\Af^\kin(\Mb;O_1)$ and $\Af^\kin(\Mb;O_2)$ commute elementwise
and false otherwise. 
We might say that $\Af$ satisfies Einstein causality for $\langle O_1,O_2\rangle$. It is easily seen that there are relationships
between these propositions:
\begin{description}
\item[R1] for all $\langle O_1,O_2\rangle\in\OO^{(2)}(\Mb)$,  
\[
P_\Mb(O_1,O_2) \iff P_\Mb(D_\Mb(O_1),D_\Mb(O_2)).
\] 
\item[R2] given $\psi:\Mb\to\Nb$ then, for all $\langle O_1,O_2\rangle\in\OO^{(2)}(\Mb)$, \[P_\Mb(O_1,O_2) \iff P_\Nb(\psi(O_1),\psi(O_2)).
\]
\item[R3] for all $\langle O_1,O_2\rangle\in\OO^{(2)}(\Mb)$
and all $\widetilde{O}_i\in\OO(\Mb)$ with $\widetilde{O}_i\subset O_i$ ($i=1,2$)
\[
P_\Mb(O_1,O_2) \implies P_\Mb(\widetilde{O}_1,\widetilde{O}_2) .
\]
\end{description}
R3 is an immediate consequence of isotony, and R1 follows
from the timeslice property.  Property R2
follows from the covariance property~\eqref{eq:kincov} of the kinematic net, 
which gives 
\begin{equation}
[\Af(\Nb;\psi(O_1)),\Af(\Nb;\psi(O_2))] = \Af(\psi)([\Af(\Mb;O_1),\Af(\Mb;O_2)])
\end{equation}
and the required property holds because $\Af(\psi)$ is injective.
In general, we will describe any collection of logical propositions
$P_\Mb:\OO^{(2)}(\Mb)$ obeying R1--R3 (with $\Mb$ varying over $\Loc$ and $\langle O_1,O_2\rangle\in\OO^{(2)}(\Mb)$) will be called \emph{rigid}. 
 
\begin{thm}\label{thm:rigidity}
Suppose $(P_\Mb)_{\Mb\in\Loc}$ is rigid, and that $P_\Mb(O_1,O_2)$
holds for some $\langle O_1,O_2\rangle\in\OO^{(2)}(\Mb)$. 
Then $P_{\widetilde{\Mb}}(\widetilde{O}_1,\widetilde{O}_2)$ 
for every $\langle\widetilde{O}_1,\widetilde{O}_2\rangle\in\OO^{(2)}(\widetilde{\Mb})$ in every spacetime $\widetilde{\Mb}\in\Loc$ for which either  (a) the Cauchy surfaces of $\widetilde{O}_i$ are oriented diffeomorphic to those of $O_i$ for $i=1,2$; or 
(b) each component of $\widetilde{O}_1\cup \widetilde{O}_2$ has Cauchy surface topology $\RR^{n-1}$.\footnote{For example, 
these components might be Cauchy developments of sets that are diffeomorphic to a $(n-1)$-ball and which lie in a spacelike Cauchy surface.}
\end{thm}
\begin{proof} The strategy for (a) is illustrated by Fig.~\ref{fig:rigidity},
in which the wavy line indicates a sequence of spacetimes
forming a deformation chain (cf.~Fig.~\ref{fig:deform})
\begin{equation}
\widetilde{\Mb}|_{\widetilde{O}_1\cup\widetilde{O}_2}
\xlongleftarrow{\widetilde{\psi}} \widetilde{\Lb} \xlongrightarrow{\widetilde{\varphi}} \Ib \xlongleftarrow{\varphi} \Lb 
\xlongrightarrow{\psi} \Mb|_{O_1\cup O_2}
\end{equation}
where $\psi,\widetilde{\psi},\varphi,\widetilde{\varphi}$ are Cauchy morphisms. By property R2, $P_\Mb(O_1,O_2)$ is equivalent to
$P_{\Mb|_{O_1\cup O_2}}(O_1,O_2)$, and likewise 
$P_{\widetilde{\Mb}}(\widetilde{O}_1,\widetilde{O}_2)$ is
equivalent to $P_{\widetilde{\Mb}|_{\widetilde{O}_1\cup \widetilde{O}_2}}(\widetilde{O}_1,\widetilde{O}_2)$. 
Writing $L_i$ and $I_i$ for the components of $\Lb$ and $\Ib$ corresponding to $O_1$ and $O_2$, and 
applying R1 and R2 repeatedly,  
\begin{align}
P_{\Mb|_{O_1\cup O_2}}(O_1,O_2)  & \xLongleftrightarrow{R1}
P_{\Mb|_{O_1\cup O_2}}(\psi(L_1),\psi(L_2))\\
& 
\xLongleftrightarrow[\psi]{R2} P_{\Lb}(L_1, L_2) 
\xLongleftrightarrow[\varphi]{R2} P_{\Ib}(\varphi(L_1), \varphi(L_2))  \xLongleftrightarrow{R1} P_{\Ib}(I_1, I_2)  \nonumber
\end{align}
and in the same way, $P_{\Ib}(I_1, I_2)$ is also equivalent to
$P_{\widetilde{\Mb}|_{\widetilde{O}_1\cup \widetilde{O}_2}}(\widetilde{O}_1,\widetilde{O}_2)$. Together with the equivalences
noted already, this completes the proof. 

For (b), we choose, for each $i=1,2$, a globally hyperbolic set $D_i$ contained in $O_i$
and with the same number of components as $\widetilde{O}_i$, and so that all its
components have Cauchy surface topology $\RR^{n-1}$. Using R3, $P_\Mb(D_1,D_2)$, and the result
follows by part (a).
\end{proof}

\begin{figure}
\begin{center}
\begin{tikzpicture}[scale=0.6]
\definecolor{Green}{rgb}{0,.80,.20}
\definecolor{Gold}{rgb}{.93,.82,.24}
\definecolor{Orange}{rgb}{1,0.5,0} 
\draw[fill=lightgray] (-5.5,0) -- ++(3.5,0) -- ++(0,4) -- ++(-3.5,0) -- cycle;
\draw[fill=lightgray] (3.5,0) -- ++(3.5,0) -- ++(0,4) -- ++(-3.5,0) -- cycle;

\draw[fill=Green] (-4.5,3) +(-0.5,0) -- +(0,-0.5) -- +(0.5,0) -- +(0,0.5) -- cycle;
\draw[fill=Green] (-3,3) +(-0.5,0) -- +(0,-0.5) -- +(0.5,0) -- +(0,0.5) -- cycle;

\draw[fill=Green] (-0,3) +(-0.5,0) -- +(0,-0.5) -- +(0.5,0) -- +(0,0.5) -- cycle;
\draw[fill=Green] (1.5,3) +(-0.5,0) -- +(0,-0.5) -- +(0.5,0) -- +(0,0.5) -- cycle;

\draw [decorate,decoration=snake] (-0.5,2) -- (2,2);

\draw[fill=Orange] (0,1) +(-0.5,0) -- +(0,-0.5) -- +(0.5,0) -- +(0,0.5) -- cycle;
\draw[fill=Orange] (1.5,1) +(-0.5,0) -- +(0,-0.5) -- +(0.5,0) -- +(0,0.5) -- cycle;

\draw[fill=Orange] (4.5,1) +(-0.5,0) -- +(0,-0.5) -- +(0.5,0) -- +(0,0.5) -- cycle;
\draw[fill=Orange] (6,1) +(-0.5,0) -- +(0,-0.5) -- +(0.5,0) -- +(0,0.5) -- cycle;

\draw[color=blue,line width=3pt,->] (2.25,1) -- (3.25,1); 
\draw[color=blue,line width=3pt,->] (-0.75,3) -- (-1.75,3); 
\node[anchor=north] at (5.5,0) {${\widetilde{\Mb}}$};
\node[anchor=north] at (-4,0) {${\Mb}$};
\node[anchor=north] at (-4.5,2.5) {${O_1}$};
\node[anchor=north] at (-3,2.5) {${O_2}$};
\node[anchor=south] at (4.5,1.5) {${\widetilde{O}_1}$};
\node[anchor=south] at (6,1.5) {${\widetilde{O}_2}$};
\node[anchor=south] at (0.75,3.5) {$\Mb|_{O_1\cup O_2}$};
\node[anchor=north] at (0.75,0) {$\widetilde{\Mb}|_{\widetilde{O_1}\cup \widetilde{O_2}}$};
\end{tikzpicture}
\end{center}
\caption{Schematic representation of the rigidity argument.}
\label{fig:rigidity}
\end{figure}
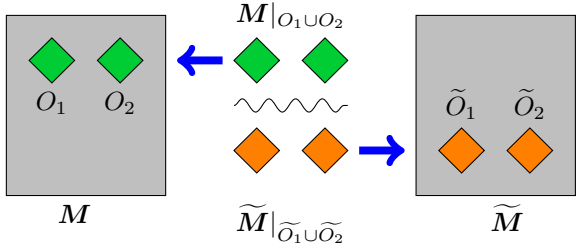

As a consequence, we see that the hypothesis that Einstein causality holds in one spacetime is not independent of it holding in another. 
This is a prototype for the spin--statistics connection that will
be described later, and is similar to the arguments used in~\cite{Verch01}. Related arguments apply to properties
such as  extended locality (see \cite{Schoch1968,Landau1969} for
the original definition) and the Schlieder property (see, likewise \cite{Schlieder:1969}) as described in~\cite{FewsterVerch_chapter:2015}.

\section{Framed spacetimes}\label{sect:FLoc}

The conventional account of theories with spin is phrased in terms of spin structures. Four dimensional globally hyperbolic spacetimes support a
unique spin bundle (up to equivalence) namely the trivial right-principal bundle $S\Mb=\Mb\times\SL(2,\CC)$~\cite{Isham_spinor:1978} and for simplicity we restrict to this situation. A spin structure $\sigma$
is a double cover from $S\Mb$ to the frame bundle $F\Mb$ over $\Mb$ that intertwines
the right-actions on $S\Mb$ and $F\Mb$: i.e., $\sigma\circ R_S =R_{\pi(S)}\circ \sigma$,
where $\pi:\SL(2,\CC)\to\Lc^\uparrow_+$ is the usual double cover. Pairs 
$(\Mb,\sigma)$ form the objects of a category $\SpLoc$, in which   
a morphism $\Psi:(\Mb,\sigma)\to (\Mb',\sigma')$ is a bundle morphism 
$\Psi:S\Mb\to S\Mb'$ which covers a $\Loc$-morphism $\psi:\Mb\to\Mb'$, i.e., 
$\Psi(p,S)=(\psi(p),\Xi(p)S)$ for some $\Xi\in C^\infty(\Mb;\SL(2,\CC))$, and obeys
$\sigma'\circ\Psi = \psi_*\circ\sigma$, where $\psi_*$ is the induced map of frame bundles 
arising from the tangent map of $\psi$. These structures provide the setting for the locally
covariant formulation of the Dirac field~\cite{Sanders_dirac:2010}, for instance. 
From an operational perspective, however, this account of spin  
it is not completely satisfactory,
because the morphisms are described at the level of the spin bundle, to which we do not have observational access, and are only fixed up to sign by the geometric map of spacetime manifolds.
To some extent, one has also introduced the understanding of spin by hand, as well, 
although this is reasonable enough when formulating specific models such as the Dirac field.

By contrast, the approach described here has a more operationally
satisfactory basis.  Instead of $\Loc$ or $\SpLoc$, we  work 
on a category of \emph{framed spacetimes} $\FLoc$ defined as follows.
An object of $\FLoc$ is a pair $\Mbb=(\Mc,e)$ where $\Mc$ is a smooth manifold
of fixed dimension $n$ on which $e=(e^\nu)_{\nu=0}^{n-1}$
is a global smooth coframe (i.e., an $n$-tuple of smooth everywhere linearly independent $1$-forms)
subject to the condition that $\Mc$, equipped with the metric, orientation and time-orientation
induced by $e$, is a spacetime in $\Loc$, to be denoted $\Lf(\Mc,e)$.
Here, the metric induced by $e$ is $\eta_{\mu\nu} e^\mu e^\nu$,
where $\eta=\text{diag}(+1,-1,\ldots,-1)$, while the orientation and time-orientation are fixed
by requiring $e^0\wedge\cdots \wedge e^{n-1}$ to be positively oriented, and $e^0$ to be future-directed. Similarly, 
a morphism $\psi:(\Mc,e)\to (\Mc',e')$ in $\FLoc$ is a smooth map between
the underlying manifolds inducing a $\Loc$-morphism $\Lf(\Mc,e)\to \Lf(\Mc',e')$ and obeying $\psi^*e' = e$. In this way, we obtain a forgetful functor $\Lf:\FLoc\to\Loc$. Moreover, $\FLoc$ is related to $\SpLoc$ by a  functor 
$\Sf:\FLoc\to\SpLoc$ defined by
\begin{equation}
\Sf(\Mc,e)=(\Lf(\Mc,e),(p,S)\mapsto R_{\pi(S)} e|_p^*),
\end{equation}
where $e|_p^*$ is the dual frame to $e$ at $p$,
and so that each $\FLoc$ morphism is mapped to a $\SpLoc$-morphism 
$\Sf(\psi)$ whose underlying bundle map is $\psi\times\id_{\SL(2,\CC)}$.
Essentially, $\Sf(\Mc,e)$ 
corresponds to the trivial spin structure associated to a frame~\cite{Isham_spinor:1978}, and we exploit the uniqueness of this spin structure to define the morphisms.  One may easily
see that $\Sf$ is a bijection on objects; however, there are morphisms in $\SpLoc$ that do not have precursors in $\FLoc$, which
involve local frame rotations.\footnote{Local frame rotations will
appear later on, but not as morphisms.}  Clearly, the composition of $\Sf$
with the obvious forgetful functor from $\SpLoc$ to $\Loc$ gives the functor $\Lf:\FLoc\to\Loc$.

The description of spacetimes in $\Loc$ represents a 
`rods and clocks' account of measurement.\footnote{One might be concerned that the assumption that global coframes exist is restrictive, as
it requires that $\Mc$ to be parallelizable. However, this presents no difficulties if $n=4$, because all four dimensional globally hyperbolic spacetimes are parallelizable. Conceivably, one could modify the set-up in general dimensions by working with local coframes, if it was felt necessary to include non-parallelizable spacetimes.}  However, we need to be clear that the coframe is not in itself physically significant,
by contrast to the metric, orientation and time-orientation it induces. In other
words, our description contains redundant information and we must take care
to account for the degeneracies we have introduced. This is not a bug, but a feature: it turns out
to lead to an enhanced understanding of what spin is.

In this new context, a locally covariant QFT should be a functor from $\FLoc$ to $\Alg$ (or some other category, e.g., $\CAlg$). 
Of course, any theory $\Af:\Loc\to\Alg$ induces such a functor, namely $\Af\circ\Lf:\FLoc\to\Alg$, 
and likewise every $\Bf:\SpLoc\to\Alg$ induces $\Bf\circ\Sf:\FLoc\to\Alg$, but not every theory need arise in this way. However, we need
to keep track of the redundancies in our description, namely the freedom to make \emph{global frame rotations}. These are represented as follows. 
To each $\Lambda\in \Lc^\uparrow_+$, there is a functor
$\Tf(\Lambda):\FLoc\to\FLoc$ 
\begin{equation}
\Tf(\Lambda)(\Mc,e) = (\Mc,\Lambda e), \qquad\text{where
$(\Lambda e)^\mu=
\Lambda^\mu_{\phantom{\mu}\nu} e^\nu$} \qquad(\Lambda\in\Lc^\uparrow_+)
\end{equation}
with action on morphisms uniquely fixed so that $\Lf\circ\Tf(\Lambda)=\Lf$. 
In this way, $\Lambda\mapsto\Tf(\Lambda)$ faithfully represents $\Lc^\uparrow_+$ in $\Aut(\FLoc)$. Moreover, any locally covariant theory
$\Af:\FLoc\to\Alg$, induces a family of theories
\begin{equation}
\Af\circ\Tf(\Lambda):\FLoc\to\Alg\qquad (\Lambda\in\Lc^\uparrow_+)
\end{equation}
which corresponds to applying the original theory $\Af$ to a frame-rotated 
version of the original spacetime. If we are to take seriously the idea
that frame rotations of this type carry no physical significance then these
theories should be equivalent. We formalise this in the following

\begin{axiom}[Independence of global frame rotations] \label{ax:global}
To each $\Lambda\in\Lc^\uparrow_+$, there exists an equivalence
$\eta(\Lambda): \Af\nto\Af\circ\Tf(\Lambda)$, such that
\begin{equation}\label{eq:gauge}
\eta(\Lambda)_{(\Mc,e)}\alpha_{(\Mc,e)}= 
\alpha_{(\Mc,\Lambda e)} \eta(\Lambda)_{(\Mc,e)} \qquad
(\forall\alpha\in\Aut(\Af))
\end{equation}
\end{axiom}
The condition \eqref{eq:gauge} asserts that the
equivalence implementing independence of global frame
rotations intertwines the action of global gauge transformations. 
Plausibly it might be relaxed (or modified) but it gives the cleanest
results, so will be maintained for now. Note that the 
equivalences $\eta(\Lambda)$ are not specified beyond this
requirement; what is important is that they exist. 
Obviously every theory induced from $\Loc$ (i.e., 
$\Af=\Bf\circ\Lf$, for some $\Bf:\Loc\to\Alg$)
obeys Axiom~\ref{ax:global}, simply by taking 
$\eta(\Lambda)$ to be the identity automorphism of $\Af$. 

The assumptions above have a number of consequences~\cite{Few_spinstats}. First, the $\eta(\Lambda)$
induce a $2$-cocycle of $\Lc^\uparrow_+$, 
taking values in the centre of the global gauge group
$\Zc(\Aut(\Af))$, and given by
\begin{equation}
\xi(\Lambda',\Lambda)_{(\Mc,e)} = \eta(\Lambda)^{-1}_{(\Mc,e)} \eta(\Lambda')^{-1}_{(\Mc,\Lambda e)} \eta(\Lambda'\Lambda)_{(\Mc,e)};
\end{equation}
furthermore, any other system of equivalences $\widetilde{\eta}(\Lambda): \Af\nto\Af\circ\Tf(\Lambda)$ obeying \eqref{eq:gauge} determines an equivalent
$2$-cocycle. We conclude that each theory $\Af:\FLoc\to\Alg$ obeying
Axiom~\ref{ax:global} determines a group cohomology class  $[\xi]\in H^2(\Lc^\uparrow_+;\Zc(\Aut(\Af)))$ in a canonical fashion. 

It is worth pausing to consider some sufficient conditions
for $[\xi]$ to be trivial. This occurs, for instance, whenever
$\Af$ is induced from a theory on $\Loc$, because 
we \emph{may} take $\eta(\Lambda)=\id_\Af$, giving 
$\xi(\Lambda, \Lambda')=\id_\Af$, and
any other choice gives an equivalent cohomologous $2$-cocycle.
Again, if $\Af$ has global gauge group with trivial centre, 
then $\xi$ has no choice but to be trivial. 

Next, the \emph{scalar} fields of the theory
form a vector space $\Fld(\Af)$ carrying an action of both the gauge group
\begin{equation}
(\alpha\cdot\Phi)_{(\Mc,e)} (f) = \alpha_{(\Mc,e)}\Phi_{(\Mc,e)}(f)
\qquad (\alpha\in\Aut(\Af))
\end{equation}
and the proper orthochronous Lorentz group $\Lc^\uparrow_+$
\begin{equation}
(\Lambda\star \Phi)_{(\Mc,\Lambda e)}(f) = \eta(\Lambda)_{(\Mc,e)}\Phi_{(\Mc,e)}(f)
\qquad (\Lambda\in\Lc^\uparrow_+).
\end{equation}
These two actions commute, and turn out to obey
\begin{equation}
(\Lambda'\Lambda)\star\Phi= \xi(\Lambda',\Lambda)\cdot
(\Lambda'\star(\Lambda\star\Phi)),
\end{equation}
which entails that irreducible subspaces of $\Fld(\Af)$ under the
action of $\Lc^\uparrow_+\times \Aut(\Af)$ carry multiplier representations of $\Lc^\uparrow_+$,  determined
by $\xi$. We deduce that the scalar fields form Lorentz and gauge multiplets (extending a result on gauge multiplets from~\cite{Fewster:gauge}). Further, all
multiplets in which the multiplier representation is continuous
(at least near the identity) must arise from
 \emph{true} real linear representations
of the covering group $\text{SL}(2,\CC)$, and are therefore
classified in the familiar way by pairs $(j,k)$ where $j,k$ are integer or half-integer spins. 
Accordingly our analysis has led to an emergent understanding of spin, 
and answers the question of why this is an appropriate 
physical notion in curved spacetimes. 

In certain cases, we may say more immediately. Any
theory induced from $\Loc$, or in which $\Zc(\Aut(\Af))$ is
trivial, can only support fields of integer-spin, because
$[\xi]$ is trivial. Similarly, all multiplets of observable fields are of integer spin, because $\xi$ is a global gauge transformation, and therefore
acts trivially on such fields. 

It seems remarkable that so much can be extracted from
the single Axiom~\ref{ax:global}, without the need to
specify what the equivalences $\eta(\Lambda)$ actually are. 
In order to prove the spin-statistics connection, however, 
it is convenient to be a bit more specific, and to connect
them to dynamics. This requires a generalization of the spacetime deformation techniques to $\FLoc$~\cite{Few_spinstats}. 

\begin{figure}
\begin{center}
\begin{tikzpicture}[scale=0.7]
\definecolor{Green}{rgb}{0,.80,.20}
\definecolor{Gold}{rgb}{.93,.82,.24}
\definecolor{Orange}{rgb}{1,0.5,0}
\draw[fill=lightgray] (-5,0) -- ++(2,0) -- ++(0,4) -- ++(-2,0) -- cycle;
\draw[fill=lightgray] (5,0) -- ++(2,0) -- ++(0,4) -- ++(-2,0) -- cycle;
\draw[fill=Gold] (5,3) -- ++(2,0) -- ++(0,0.5) -- ++(-2,0) -- cycle;
\draw[fill=Gold] (0,3) -- ++(2,0) -- ++(0,0.5) -- ++(-2,0) -- cycle;
\draw[fill=Gold] (-5,3) -- ++(2,0) -- ++(0,0.5) -- ++(-2,0) -- cycle;
\draw[fill=Gold] (5,0.5) -- ++(2,0) -- ++(0,0.5) -- ++(-2,0) -- cycle;
\draw[fill=Gold] (0,0.5) -- ++(2,0) -- ++(0,0.5) -- ++(-2,0) -- cycle;
\draw[fill=Gold] (-5,0.5) -- ++(2,0) -- ++(0,0.5) -- ++(-2,0) -- cycle;
\draw[color=red,thick] (5,3.25) -- ++(2,0);
\draw[color=red,thick] (5,0.75) -- ++(2,0);
\draw[color=red,thick] (-5,3.25) -- ++(2,0);
\draw[color=red,thick] (-5,0.75) -- ++(2,0);
\draw[color=blue,line width=4pt,->] (2.25,3.25) -- (4.75,3.25) node[pos=0.4,above]{$\iota^+[\widetilde{\Lambda}]$};
\draw[color=blue,line width=4pt,->] (2.25,0.75) -- (4.75,0.75) node[pos=0.4,above]{$\iota^-[\widetilde{\Lambda}]$};
\draw[color=blue,line width=4pt,->] (-0.25,3.25) -- (-2.75,3.25) node[pos=0.4,above]{$\iota^+$};
\draw[color=blue,line width=4pt,->] (-0.25,0.75) -- (-2.75,0.75) node[pos=0.4,above]{$\iota^-$};
\draw[fill=Orange] (6,2) ellipse (0.7 and 0.4);
\node at (6,2) {$\widetilde{\Lambda}$};
\node[anchor=north] at (6,0) {${(\Mc,\widetilde{\Lambda}e)}$};
\node[anchor=north] at (-4,0) {${(\Mc,e)}$};
\node[anchor=north] at (1,3) {${(\Mc^+,e)}$};
\node[anchor=north] at (1,0.5) {${(\Mc^-,e)}$};
\end{tikzpicture}
\end{center}
\caption{Schematic representation of the relative Cauchy 
evolution induced by a local frame rotation.}\label{fig:localrot}
\end{figure}
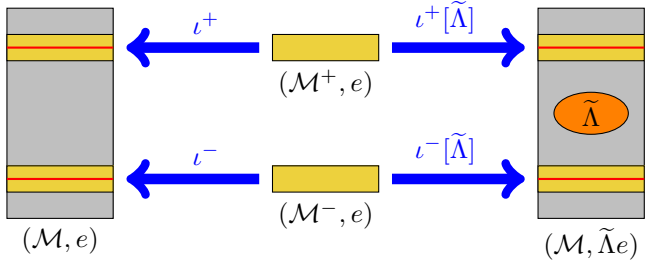

With this in mind, let us define $\FLoc$-Cauchy morphisms to
be $\FLoc$ morphisms $\psi$ whose image $\Lf(\psi)$ in $\Loc$
is Cauchy according to our earlier definition. Further, 
let us assume that $\Af:\FLoc\to\Alg$ has the timeslice property and so 
maps any $\FLoc$-Cauchy morphism to an isomorphism in $\Alg$. 
Fixing $(\Mc,e)\in\FLoc$, any $\widetilde{\Lambda}\in C^\infty(\Mc;\Lc^\uparrow_+)$ that is trivial outside a time-compact  set\footnote{That is, a set that lies to the
future of one Cauchy surface and the past of another.}
induces a relative Cauchy evolution, illustrated in Fig.~\ref{fig:localrot},
and given by
\begin{equation}
\rce_{(\Mc,e)}[\widetilde{\Lambda}] =
\Af(\iota^-)\circ\Af(\iota^-[\widetilde{\Lambda}])^{-1}
\circ\Af(\iota^+[\widetilde{\Lambda}])\circ\Af(\iota^+)^{-1}.
\end{equation} 
However, it would seem strange if a local frame rotation,
the effect of which is removed, could induce physical effects. 
Taking a more conservative stance, let us weaken that to
cover only frame rotations that can be deformed away
homotopically. It seems reasonable to posit:

\begin{axiom}[Independence of \alert{local} frame rotations]
\label{ax:local}
$\rce_{(\Mc,e)}[\widetilde{\Lambda}]=\id_{\Af(\Mc,e)}$  
for homotopically trivial $\widetilde{\Lambda}$. 
\end{axiom}

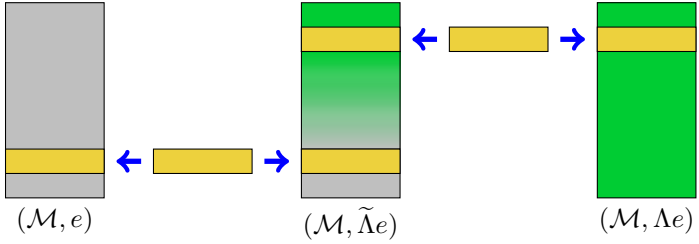
\begin{figure}
\begin{center}
\begin{tikzpicture}[scale=0.65]

\definecolor{Green}{rgb}{0,.80,.20}
\definecolor{Gold}{rgb}{.93,.82,.24}
\definecolor{Orange}{rgb}{1,0.5,0}
\draw[fill=lightgray] (-6,0) -- ++(2,0) -- ++(0,4) -- ++(-2,0) -- cycle;
\draw[fill=Gold] (-6,0.5) -- ++(2,0) -- ++(0,0.5) -- ++(-2,0) -- cycle;

\draw[fill=Green] (6,0) -- ++(2,0) -- ++(0,4) -- ++(-2,0) -- cycle;
\draw[fill=Gold] (6,3) -- ++(2,0) -- ++(0,0.5) -- ++(-2,0) -- cycle;

\draw[fill=lightgray] (0,0) -- ++(2,0) -- ++(0,1) -- ++(-2,0) -- cycle;
\draw[top color=Green,bottom color=lightgray] (0,1) -- ++(2,0) -- ++(0,2) -- ++(-2,0) -- cycle;
\draw[fill=Green] (0,3) -- ++(2,0) -- ++(0,1) -- ++(-2,0) -- cycle;
\draw[fill=Gold] (0,3) -- ++(2,0) -- ++(0,0.5) -- ++(-2,0) -- cycle;
\draw[fill=Gold] (0,0.5) -- ++(2,0) -- ++(0,0.5) -- ++(-2,0) -- cycle;

\draw[fill=Gold] (3,3) -- ++(2,0) -- ++(0,0.5) -- ++(-2,0) -- cycle;
\draw[fill=Gold] (-3,0.5) -- ++(2,0) -- ++(0,0.5) -- ++(-2,0) -- cycle;

\draw[color=blue,line width=2pt,->] (2.75,3.25) -- (2.25,3.25); 
\draw[color=blue,line width=2pt,->] (-0.75,0.75) -- (-0.25,0.75); 
\draw[color=blue,line width=2pt,->] (5.25,3.25) -- (5.75,3.25); 
\draw[color=blue,line width=2pt,->] (-3.25,0.75) -- (-3.75,0.75); 
\node[anchor=north] at (7,0) {${(\Mc,\Lambda e)}$};
\node[anchor=north] at (1,0) {${(\Mc,\widetilde{\Lambda}e)}$};
\node[anchor=north] at (-5,0) {${(\Mc,e)}$};
\end{tikzpicture}
\end{center}
\caption{Construction of the natural transformations $\zeta(S)$.}
\label{fig:rot}
\end{figure}

Axiom~\ref{ax:local} has an important consequence. 
Consider the chain of spacetimes illustrated in Fig.~\ref{fig:rot},
in which the morphisms illustrated are all Cauchy,
and $\widetilde{\Lambda}\in C^\infty(\Mc;\Lc^\uparrow_+)$ is equal to the identity in the past region 
and takes the constant value $\Lambda$ in the future region. 
Then the timeslice axiom induces an isomorphism 
$\Af(\Mc,e)\to\Af(\Mc,\Lambda e)$. Crucially, Axiom~\ref{ax:local} entails that the isomorphism
depends on $\widetilde{\Lambda}$ only via its homotopy class. 
Thus each $S$ in the universal cover $\widetilde{\Lc^\uparrow_+}$ of $\Lc^\uparrow_+$  induces isomorphisms
\begin{equation}
\zeta_{(\Mc,e)}(S): 
\Af(\Mc,e)\longrightarrow \Af(\Mc,\pi(S) e).
\end{equation}
Let us assume (although one might suspect this can be \emph{derived})
that the $\zeta_{(\Mc,e)}(S)$ cohere to give natural isomorphisms 
\begin{equation}
\zeta(S):\Af\nto \Af\circ\Tf(\pi(S)).
\end{equation}
We may now replicate our previous analysis, with $S\mapsto\zeta(S)$
in place of $\Lambda\mapsto\eta(\Lambda)$, leading to a $2$-cocycle of the universal cover of $\Lc^\uparrow_+$ in $\Aut(\Af)$ that is
trivial; indeed, one may show that 
\[
\zeta(S')_{(\Mc,\pi(S)e)} \zeta(S)_{(\Mc, e)}= \zeta(S'S)_{(\Mc, e)}
\qquad (S,S'\in\Lc^\uparrow_+).
\] 
In $n=4$ dimensions, we note that $\zeta(-1)$ is an automorphism
of $\Af$ (as $\pi(-1)=1$); moreover, it obeys  
\begin{equation}  
\zeta(-1)^2 = \zeta(1)= \id,
\end{equation}
which one might think of as a spacetime version of Dirac's belt trick.

It is important to connect our discussion of frame rotations with the familiar
implementation of the Lorentz group in Minkowski space. In our present setting, 
we define $n$-dimensional Minkowski space to be the object $\Mbb_0=(\RR^n,(dX^\mu)_{\mu=0}^{n-1})$, where $X^\mu:\RR^n\to\RR$ are the standard coordinate functions $X^\mu(x^0,\ldots,x^{n-1})=x^\mu$. Any 
$\Lambda\in \Lc^\uparrow_+$ induces an active Lorentz transformation $\Lambda:\RR^4\to\RR^4$
by matrix multiplication, $X^\mu \circ \Lambda =\Lambda^\mu_{\phantom{\mu}\nu}X^\nu$, 
which induces a morphism
\begin{equation}
\psi_\Lambda: \Mbb_0\to \Tf(\Lambda^{-1})(\Mbb_0)
\end{equation}
in $\FLoc$. One may verify that $\psi_{\Lambda'\Lambda}=\Tf(\Lambda^{-1})(\psi_{\Lambda'})\circ\psi_\Lambda$. Accordingly, we obtain an automorphism of
$\Af(\Mbb_0)$ for each $S\in\widetilde{\Lc^\uparrow_+}$ by 
\begin{equation}
\Xi(S)=\zeta(S)_{\Tf(\pi(S)^{-1})(\Mbb_0)}\circ\Af(\psi_{\pi(S)}).
\end{equation}
It may be checked that $\Xi(S'S)=\Xi(S')\circ\Xi(S)$ and that one has
\begin{equation}
\Xi(S)\Phi_{\Mbb_0}(f) = (S\star \Phi)_{\Mbb_0}(\pi(S)_* f)\qquad(f\in\CoinX{\RR^4})
\end{equation}
where we now extend the action on fields from the Lorentz group to its universal cover. 
In particular, note that any $2\pi$-rotation corresponds to
\begin{equation}\label{eq:twopi}
\Xi(-1)=\zeta(-1)_{\Mbb_0}.
\end{equation} 
Given a state $\omega_0$ on $\Af(\Mbb_0)$ that is invariant under the automorphisms, 
i.e., $\omega_0\circ\Xi(S)=\omega_0$ for all $S$, the corresponding GNS representation
will carry a unitary implementation of the $\Xi(S)$, which recovers the standard
formulation.

\section{Spin and Statistics in four dimensions}

We come to the proof of the spin--statistics connection~\cite{Few_spinstats}. 
As in \cite{Verch01}, the idea is to refer the statement in 
a general spacetime back to Minkowski space, where
standard spin--statistics results can be applied. In other
words, we apply a rigidity argument. The notion of
statistics employed is based on graded commutativity of local algebras at
spacelike separation. 

\begin{definition} An involutory global gauge transformation  $\gamma\in\Aut(\Af)$, $\gamma^2=\id$ is
 said to \emph{grade statistics} in $\Mbb$ if, for all spacelike separated regions $O_i\in\OO(\Mbb)$,
every component of which has Cauchy surface topology $\RR^3$, one has
\begin{equation}
A_1 A_2 = (-1)^{\sigma_1\sigma_2} A_2 A_1
\end{equation}
for all $A_i\in\Af^\kin(\Mbb;O_i)$ s.t., $\gamma_{\Mbb} A_i=(-1)^{\sigma_i}A_i$. 
\end{definition}

The standard spin--statistics connection, in view of \eqref{eq:twopi},
asserts that $\zeta(-1)$ grades statistics in Minkowski space $\Mbb_0$, where $\zeta(S)$ is defined
as in Sect.~\ref{sect:FLoc}. 

\begin{thm} 
If $\gamma$ grades statistics in $\Mbb_0$, then
it does so in every spacetime of $\FLoc$. 
Consequently, if the theory obeys the standard spin--statistics connection in Minkowski space, $\zeta(-1)$ grades statistics on every framed spacetime $\Mbb\in\FLoc$. 
\end{thm}
\begin{proof} (Sketch)
For each $\langle O_1,O_2\rangle\in\OO^{(2)}(\Mb)$, let $P_{\Mbb}(O_1,O_2)$ be the statement that
\begin{equation}
A_1 A_2 = (-1)^{\sigma_1\sigma_2} A_2 A_1 \qquad 
\text{for all $A_i\in\Af^\kin(\Mbb;O_i)$ s.t., $\gamma_{\Mbb} A_i=(-1)^{\sigma_i}A_i$}
\end{equation}
We argue that the collection $(P_{\Mbb})_{\Mbb\in\FLoc}$ is rigid,
whereupon the result holds by a generalization of Theorem~\ref{thm:rigidity} to $\FLoc$. 
R1 and R3 hold for the same reasons used in Section~\ref{sect:rigidity}
for Einstein causality. For R2, we note that the subspaces 
\begin{equation}
\Af^\kin_\sigma(\Mbb;O) = \{A\in\Af^\kin(\Mbb;O):
\gamma_\Mbb A=(-1)^\sigma A\} \qquad \sigma\in\{0,1\}
\end{equation}
obey, for any $\psi:\Mbb\to\widetilde{\Mbb}$, 
\begin{equation}
\Af^\kin_\sigma(\widetilde{\Mbb};\psi(O)) =
\Af(\psi)(\Af^\kin_\sigma(\Mbb;O) )
\end{equation}
by naturality of $\gamma$ and injectivity of $\Af(\psi)$.
A further use of injectivity gives
\begin{equation}
P_{\widetilde{\Mbb}}(\psi(O_1),\psi(O_2)) \iff
P_{\Mbb}(O_1,O_2),
\end{equation} 
thus establishing R2 and concluding the proof.
\end{proof}

What is really being proved is the connection between the statistics grading
in Minkowski space and that in arbitrary spacetimes. Thus, a locally covariant
theory that violates the standard spin-statistics connection in Minkowski space (e.g., 
a ghost theory) but in which the statistics grading is still implemented (in Minkowski) by an
involutory gauge transformation, would be covered by our result - the statistics
would be consistently graded in all spacetimes by the same gauge transformation.

\section{Summary and Outlook}

The BFV paper \cite{BrFrVe03} is subtitled `A new paradigm for local quantum physics', and indeed their paper marked  
the beginnings of a full development of a model-independent account of QFT in CST,
the current state of which is described in more detail in~\cite{FewsterVerch_chapter:2015}. 
At the heart of this approach is the fact that local covariance is
a surprisingly rigid structure, which makes it possible to transfer certain
results from the flat spacetime situation into general curved spacetimes
in a fairly systematic way. This is a consequence of the timeslice
property and also the structure of the categories $\Loc$ and
$\FLoc$. 

In this contribution, I have focussed particularly on the spin-statistics
connection, which was one of the starting points for the general theory. 
I have described a new viewpoint, based on framed spacetimes, 
that gives a more operational starting point for the discussion of 
spin in locally covariant QFT, without making reference to
unobservable geometric structures such as spin bundles. 
Instead, by recognizing that we make physical measurements
using frames, and by tracking the 
concomitant redundancies, we are led naturally to a description
that allows for spin. In our discussion, the relative Cauchy evolution, 
which plays an important role in locally covariant physics on $\Loc$, 
is developed further so as to cater for deformation of the framing, 
rather than just of the metric. 

Certain issues remain to be understood. 
Our view of statistics has focussed on graded commutativity
at spacelike separation;
it is not currently clear how to make contact with 
the occurrence of braid statistics in low dimensions.
The coframed spacetimes we consider are necessarily parallelizable;
while this is not a restriction in four spacetime dimensions, one could
seek generalizations that accommodate nonparallelizable spacetimes of other dimensions.  
Finally, neither the result described here, nor Verch's result
~\cite{Verch01}, gives a direct proof of the spin-statistics
connection in curved spacetime; both rely on the classic results of Minkowski space QFT. 
Now a proof is a proof, and perhaps one should not 
complain too much, because it may be that a direct argument
would be considerably more involved than those we now have. 
Nonetheless, arguments that provide more insight into the 
nature of the spin-statistics connection are still desirable
and it is hoped that the more operational account of spin
presented here can be a further step along that path. 


\subsection*{Acknowledgment} 
I thank the organisers and participants of the \emph{Quantum Mathematical Physics}
conference in Regensburg (2014) for their interest and comments, and also the 
various sponsoring organizations of the meeting for financial support.

\end{document}